\documentclass[sigconf]{acmart}

\setcitestyle{square}

\usepackage{booktabs} 
\usepackage{blindtext}
\usepackage{multirow}
\usepackage{subfig}
\usepackage{stfloats}
\usepackage{balance}
\usepackage[linesnumbered,ruled]{algorithm2e}

\definecolor{blue}{rgb}{0,0,0}

\renewcommand{\shortauthors}{B. Trovato et al.}

\newcommand{\calm}{\mathcal{M}}
\newcommand{\calf}{\mathcal{F}}
\newcommand{\cals}{\mathcal{S}}

\newcommand{\caln}{\mathcal{N}}

\newcommand{\lnorm}{\left|\left|}
\newcommand{\rnorm}{\right|\right|}

\newcommand{\E}{\mathbb{E}}
\newcommand{\Gmi}{\mathcal{G}_{-i}}
\newcommand{\Gmiv}{ G_{-i}}

\copyrightyear{2019}
\acmYear{2019}
\setcopyright{acmlicensed}
\acmConference[ETRA '19]{2019 Symposium on Eye Tracking Research and
Applications}{June 25--28, 2019}{Denver , CO, USA}
\acmBooktitle{2019 Symposium on Eye Tracking Research and Applications
(ETRA '19), June 25--28, 2019, Denver , CO, USA}
\acmPrice{15.00}
\acmDOI{10.1145/3314111.3319823}
\acmISBN{978-1-4503-6709-7/19/06}

\acmSubmissionID{57}
\title{Differential Privacy for Eye-Tracking Data}
\begin{document}


\author{Ao Liu}
\orcid{1234-5678-9012}
\affiliation{%
  \institution{Rensselaer Polytechnic Institute}
  \streetaddress{110 8th Street}
  \city{Troy}
  \state{New York}
  \postcode{12180}
}
\email{liua6@rpi.edu}

\author{Lirong Xia}
\orcid{1234-5678-9012}
\affiliation{%
  \institution{Rensselaer Polytechnic Institute}
  \streetaddress{110 8th Street}
  \city{Troy}
  \state{New York}
  \postcode{12180}
}
\email{xial@cs.rpi.edu}

\author{Andrew Duchowski}
\orcid{1234-5678-9012}
\affiliation{%
  \institution{Clemson University}
  \streetaddress{Clemson University}
  \city{Clemson}
  \state{South Carolina}
  \postcode{29634}
}
\email{duchowski@clemson.edu}

\author{Reynold Bailey}
\orcid{1234-5678-9012}
\affiliation{%
  \institution{Rochester Institute of Technology}
  \streetaddress{1 Lomb Memorial Dr}
  \city{Rochester}
  \state{New York}
  \postcode{14623}
}
\email{rjb@cs.rit.edu}

\author{Kenneth Holmqvist}
\orcid{1234-5678-9012}
\affiliation{%
  \institution{University of Regensburg}
  \streetaddress{Universitätsstraße 31}
  \city{Regensburg}
  \state{Germany}
  \postcode{93053}
}
\email{Kenneth.Holmqvist@psychologie.uni-regensburg.de}

\author{Eakta Jain}
\orcid{1234-5678-9012}
\affiliation{%
  \institution{University of Florida}
  \streetaddress{University of Florida}
  \city{Gainesville}
  \state{Florida}
  \postcode{32611}
}
\email{ejain@cise.ufl.edu}

\renewcommand{\shortauthors}{Liu et al.}

\begin{abstract}
As large eye-tracking datasets are created, data privacy is a pressing concern for the eye-tracking community. De-identifying data does not guarantee privacy because multiple datasets can be linked for inferences. \textcolor{blue}{A common belief is that aggregating individuals' data into composite representations such as heatmaps protects the individual.} \textcolor{blue}{However, we analytically examine the privacy of (noise-free) heatmaps and show that they do not guarantee privacy.} We further propose two noise mechanisms that guarantee privacy and analyze their privacy-utility tradeoff. Analysis reveals that our Gaussian noise mechanism is an elegant solution to preserve privacy for heatmaps. Our results have implications for interdisciplinary research to create differentially private mechanisms for eye tracking.
\end{abstract}

%
%

\ccsdesc[500]{Security and privacy~Human and societal aspects of security and privacy}
\ccsdesc[500]{Security and privacy~Privacy protections}

\keywords{Eye-tracking, Differential Privacy, Privacy-Utility Tradeoff, Heatmaps}


\maketitle
\section{Introduction}

With advances in mobile and ubiquitous eye tracking, there is ample opportunity to collect eye tracking data at scale. A user's gaze encodes valuable information including attention, intent, emotional state, cognitive ability, and health. This information can be used to gain insight into human behavior (e.g. in marketing and user experience design), create computational models (e.g. for smart environments and vehicles), and enable interventions (e.g. health and education). When combined with physiological sensing and contextual data, this information facilitates the modeling and prediction of human behavior and decision making. As users become increasingly conscious about what their data reveals about them, there is mounting pressure on policymakers and corporations to introduce robust privacy regulations and processes \cite{gdpr:2018,usatoday:2018}. The eye tracking community must actively pursue research about privacy for broad public acceptance of this technology.


Data privacy for eye tracking has been raised as a concern in the community \cite{ling2014synergies,KAB18}. At a recent Dagstuhl seminar on ubiquitous gaze sensing and interaction\nolinebreak\footnote{%
    \url{https://www.dagstuhl.de/18252}
}, privacy considerations were highlighted in a number of papers in the proceedings \cite{CDQW18}. Privacy as a general term has a wide range of meanings and different levels of importance for different users. Privacy can obviously be preserved by distorting or randomizing the answers to queries, however doing so renders the information in the dataset useless. 


To maintain privacy while preserving the utility of the information, we propose to apply the concept of \textit{differential privacy} (DP) which has been developed by theoretical computer scientists and applied to database applications over the past decade \cite{Dwork:2011:privatedataanalysis}. Differential privacy can be summarized as follows:
\begin{quote}
 Privacy is maintained if an individual's records
 cannot be accurately identified, even in the worst case when all other data has been exposed by adversaries.
\end{quote}


\noindent Our technical contributions are:
(1) We introduce the notion of {differential privacy} for eye tracking data. (2) We formally examine the privacy of aggregating eye tracking data as heatmaps and show that aggregating into heatmaps does not guarantee privacy from a DP perspective. (3) We propose two mechanisms to improve the privacy of aggregated gaze data. (4) We analyze the privacy-utility trade-off of these mechanisms from a DP-point of view.

\begin{figure*}[tbp]
    \centering
    \includegraphics[width=0.78\textwidth]{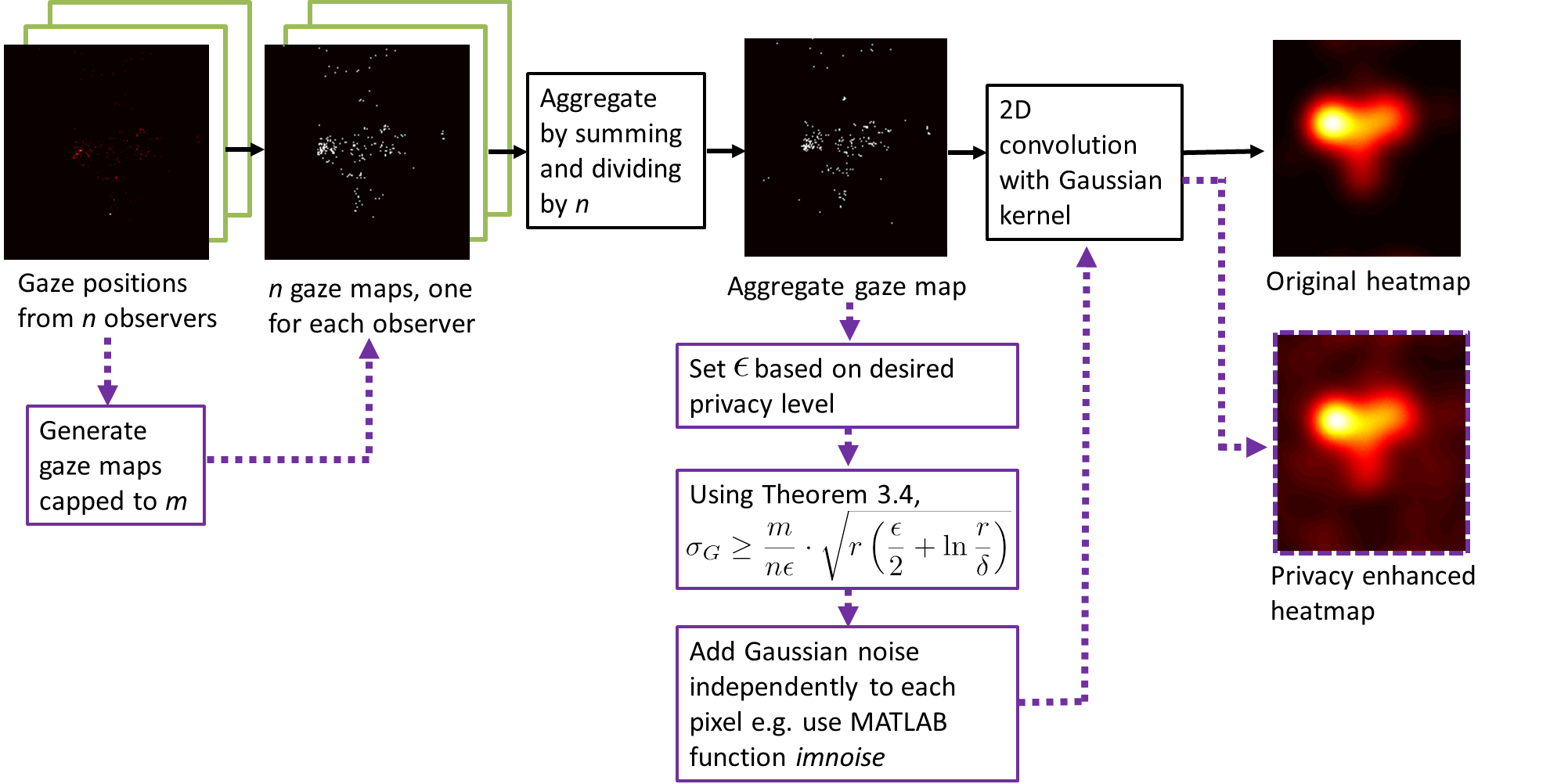}
    \caption{
    Workflow for researchers and practioners to create the desired strength of privacy level. The solid lines illustrate the standard workflow for generating an aggregate static heatmap from eye tracking data. The dotted lines show how to implement a privacy protocol with small modifications to this workflow. The hotspots on the privacy enhanced heatmap are visually in the same locations as the original heatmap. The supplementary materials show several examples of privacy enhanced heatmaps for the same noise level.}
    \label{fig:workflow}
\end{figure*}
\begin{table*}[b!p]
    \centering
    \caption{In most cases, eye tracking data is released with the stimuli. This table illustrates the threats posed by releasing this data if no privacy protocol is in place.}
    \label{tab:threats}
    \begin{tabular}{ p{3cm}|p{3.5cm}|p{3cm}|p{3cm}|p{3cm}}
    \toprule
   \textbf{Type of data} & \textbf{Example of intended use} & \textbf{What adversary can access in worst case} & \textbf{What adversary can do now} & \textbf{Does DP apply?}\\
   \midrule
    Raw eye movements & Foveated rendering & Raw eye movements  & Neurological diagnoses (see Scenario 1) & yes, future work \\
    \midrule
    Aggregated data without temporal information (static heatmaps)
    & Marketing, UX design, education & Individual's heatmap & Behavioral diagnoses (see Scenario 2) & yes, this paper \\
    \midrule
    Aggregated data with temporal information (dynamic heatmaps) & Training models for autonomous vehicles & Individual's heatmap & Establish driver's liability (see Scenario 3) & yes, future work \\
    \midrule
    Areas of Interest (AOI) analysis & Expert vs novice analysis & Individual's AOI visit order & Autism spectrum diagnoses & yes, future work\\
    \bottomrule
    \end{tabular}
\end{table*}

From a practical perspective, the notion of differential privacy is both achievable and theoretically verifiable. Though the proofs may be mathematically sophisticated, the implementation is straight\-forward, and can be integrated into the eye tracking data collection pipeline. Figure~\ref{fig:workflow} illustrates how this may be achieved. Privacy is guaranteed for the worst case when an adversary has already gained access to the data of all other individuals in a dataset (by hacking them for example). Even in this case, the adversary will still not be able to accurately infer data records of the individual. In applying the general definition of differential privacy to eye tracking, we acknowledge that individual users, service providers, and policy makers may have different positions on what level of privacy versus utility is desirable. Our work provides a theoretically grounded analysis of privacy preserving mechanisms to empower these stakeholders to make such decisions. 




\textbf{Implications.} Table~\ref{tab:threats} presents some of the threats that may be posed if an adversary was to access eye tracking data with no privacy protocol in place. Specifically, we elaborate three scenarios where eye tracking data is collected with good intentions, but if hacked, could have consequences for the individuals concerned.

Scenario 1: \textit{A hospital or doctor's office collects eye tracking data as part of patients' general examination. A research grant enables a team to use this data to build a machine learning model that can  predict whether someone has a certain neurological disorder. A hacker gains unauthorized access to this database and is able to identify specific individuals with the disorder. The hacker then sells or publicly releases the identity of these individuals, negatively impacting their employment opportunities, inflating their health insurance costs, and elevating their social and emotional anxiety.} 


Scenario 2: \textit{A parent signs a consent form allowing her child to be eye tracked in a classroom. The consent form says that this data is for a research project to understand and characterize learning disabilities and build interventions. The anonymized dataset will be released as part of an NIH big data initiative. If an adversary manages to access an individual child's data and analyze it for markers of dyslexia (for example), they may sell the information to a marketing company that will contact the parent with unsolicited advertising for therapies.}

Scenario 3: \textit{A publicly funded research team is using eye tracking to study awareness and fatigue of commercial truck drivers. The eye movement data along with the scene being viewed is streamed to a remote server for later analysis. A driver in the study was involved in an accident that resulted in a fatality. Although drivers were told their data would be de-identified, a private investigator, hired by the family of the deceased, was able to extract his/her data record from the database, revealing evidence that (s)he was at fault in the accident.}

In scenarios such as these, research teams may reassure participants that raw data will not be released, or that individual data will be de-identified or aggregated (often in the form of heatmaps), providing the impression that privacy is preserved.





\section{Background}\label{background}

\subsubsection{The problem with de-identification.}
The first ``solution'' that occurs to many of us is to simply anonymize, or 
de-identify the dataset. This operation refers to removing personal identifiers 
such as the name of the participant from the dataset.
The problem with this approach is that it is not future-proof; as newer datasets 
are released, multiple datasets can be linked, and the identity of a participant 
can then be inferred \cite{nissim2017differential,ohm2009broken,holland2011biometric,komogortsev2010biometric}.

\subsubsection{The problem with running queries.} \hspace{0.05cm} A second ``solution'' would be to not release the dataset as is, rather allow the analyst to query the dataset. The dataset would not allow queries on individual items, but only on large numbers of items. In other words, a query such as \textit{``Where did the student with the lowest grade look?''} would be disallowed. But then, the analyst can run queries such as \textit{``Where did the students who did not have the lowest grade look?''}, and \textit{``Where did all the students look?''}, and use these queries to infer the disallowed query. This ``solution'' is not able to guarantee privacy in the worst case, for example, if the adversary hacks the data of $n-1$ out of $n$ persons in the dataset. Then (s)he can easily infer the $n$th person's data by querying the average or sum of the dataset.

These issues are well known in database research. One widely accepted formal definition of privacy that has emerged from this extensive research is as follows: an individual's privacy is preserved if the inferences that are made from the dataset do not indicate in any significant way whether this individual is part of the dataset or not. This notion is called \textit{differential privacy}.

\subsubsection{Differential privacy.}
Differential privacy as a concept was conceived through insights by theoretical computer scientists aiming to formalize the notion of privacy that was practically achievable as well as theoretically verifiable \cite{Dwork:2011:privatedataanalysis}. A survey of differential privacy in different fields is presented by Dwork~\shortcite{dwork2008differential}. Relevant to eye tracking are the works that have applied differential privacy definitions to machine learning \cite{ji2014differential,abadi2016deep} and time-series analysis \cite{fan2014adaptive,rastogi2010differentially}. From a societal impact perspective, the eye tracking industry has as much to gain from these ideas.

\subsubsection{Mathematical definition of differential privacy.}
Formally, given datasets $D$ and $D'$ that differ in at most one entry,
let $\calm$ denote a randomized mechanism that outputs a query of a database with some probability. Then, let $\cals$ denote a subset of query outcomes (called an ``event''). Then, we say the mechanism $\calm$ is $\epsilon-$differentially private (or $\epsilon-$DP in short) if for any $\cals$, $D$ and $D'$,
\begin{equation}\label{equ:dp1}
    Pr\left[\calm(D) \in \cals\right] \leq e^\epsilon Pr\left[\calm(D') \in \cals\right],
\end{equation}
In the above inequality, the probability comes from the randomness of mechanism $\calm$. Such randomness is necessary as we will see in Section~\ref{subs:nofree}. We note that this is a worst-case analysis that offers a strong guarantee of privacy, because the inequality must hold for all $\cals$, and all neighboring datasets $D$ and $D'$. 

Another more applicable notion of differential privacy is $(\epsilon,\delta)-$\\differential privacy, which is a generalization of $\epsilon-$DP. Using the notation above, we say the mechanism $\calm$ is $(\epsilon,\delta)-$differentially private (or $(\epsilon,\delta)-$DP in short) if for any $\cals$, $D$ and $D'$ ($D$ and $D'$ differs at most one entry),
{
\begin{equation}\nonumber
    Pr\left[\calm(D) \in \cals\right] \leq e^\epsilon Pr\left[\calm(D') \in \cals\right] + \delta,
\end{equation}
}
Typically it is believed that $\delta = \Omega\left(\frac{1}{n}\right)$ means poor privacy \cite{dwork2014algorithmic} because it allows some individuals' data to be fully recovered, 
where $n$ is the input size. We note that a mechanism can be $(\epsilon, \delta)$-DP for multiple combinations of $(\epsilon, \delta)$. As a rule of thumb, smaller $\epsilon$'s and $\delta$'s means better privacy, though we must point out that directly comparing different numerical values is not informative, e.g.~$(0.1,0.1)$ and $(1,0)$ are not comparable.



\subsubsection{Toy example.}
As part of a general wellness dataset $D$, the heights of five people are collected. The mean value as the average height of the population is released. Here, $\cals$ is the set of outputting average height. In this example, an adversary obtains the heights of four of these five persons through hacking. In this way, the adversary has a dataset $D'$ that contains all persons except the fifth. The adversary computes the average height of the dataset $D'$ and finds that it is much lower than the average height of the dataset $D$. The adversary thus infers that the fifth person must be very tall.\footnote{%
    The adversary
    can also compute the exact height of the fifth person.
}
In other words, even though the fifth person was not known by the adversary, and the dataset $D$ was not released (only the average height was released), the fifth person is also compromised because his or her height can be reverse engineered by the adversary. Now, we introduce a mechanism $\calm$ that perturbs the average height of the dataset $D$ by a random amount before releasing it. If the level of perturbation is high enough, the adversary will not be able to even infer whether the fifth person is tall or not. Thus the mechanism $\calm$ protects the privacy of the fifth person. Of course, if we add too much perturbation (or, output totally at random), the utility of the dataset will be affected because the output average height contains little information and does not reflect the average height of the population. This is the privacy-utility tradeoff (see Section~\ref{sec:put}). 


\subsubsection{Privacy in eye tracking.}
For much of the past two decades, the focus of eye tracking research has been on making eye tracking ubiquitous, and on discovering the breadth of inferences that can be made from this data, especially in the contexts of health \cite{leigh2015neurology} and education \cite{jarodzka2017eye}. Privacy has not been a high priority because of the benefits of identifying pathology and designing personalized interventions. The relevance of privacy to eye tracking data was eloquently discussed by \citeN{liebling2014privacy}.
\citeN{ling2014synergies} and \citeN{KAB18} have also highlighted the need for eye-tracking data.
Privacy considerations have been raised both for streaming data, as well as pre-recorded datasets. Despite growing awareness and concern, few solutions have been proposed. Our work provides a technical solution for the privacy of individuals.

\subsubsection{Why heatmaps as the first for privacy analysis.}
%
Besides {\em scanpaths}, the {\em heatmap} is a popular method of visualizing
eye movement data \cite{Duc18}.
Heatmaps,
or attentional landscapes
as introduced by \citeN{pomplun1996disambiguating}
and popularized by \citeN{wooding2002fixation},
are used to represent aggregate fixations \cite{duchowski2012aggregate}.
Other similar approaches involve gaze represented as height maps
\cite{ESW84,vvA07} or Gaussian Mixture Models \cite{MSHH11}.
Heatmaps are generated by accumulating exponentially decaying intensity
$I(i,j)$ at pixel coordinates $(i,j)$ relative to a fixation at
coordinates $(x,y)$,
\begin{equation*}
        \label{eqn:gauss}
        I(i,j)  =  \exp \left(
                        -((x-i)^2 + (y-j)^2) / (2 \sigma^2)
                        \right)
\end{equation*}
where the exponential decay is modeled by the Gaussian point spread function.
A GPU-based implementation \cite{duchowski2012aggregate} is available for real-time
visualization. Though heatmaps are very popular as a visualization, AOI analyses and temporal
data analysis is key to eye-tracking research. We have focused on static heatmaps as a proof
of concept for the applicability of differential privacy (DP) to eye tracking data.
Insights from this work will inform future research on
privacy in eye tracking.
\section{Analyzing differential privacy of the proposed privacy-preserving mechanisms}\label{sec:theoretical}
In this section, we analyze the differential privacy of four natural random mechanisms. 
\textcolor{blue}{We show two of these mechanisms cannot preserve privacy under the notion of DP.
For the other two mechanisms, we provide theoretically guaranteed lower bounds on the noise required for any user-defined privacy level.}
Because a heatmap is created from aggregation of gaze maps, and because this
is a reversible (convolution) process, the privacy of a heatmap is
equivalent to that of the aggregated gaze map on which it is based.

\subsection{Notations}\label{subs:notation}
We use $n$ to denote the number of observers in the database and $r$ to denote the total number of pixels in the gaze maps. For example, an image of resolution $800\times 600$ corresponds to $r = 4.8\times 10^5$. We introduce an integer $m>1$ to cap every observer's gaze map. For example, if an observer looked at one pixel more than $m$ times, we only count $m$ in his/her gaze map.\footnote{Think of this as if the gaze map \textit{saturated}.} In Section~\ref{sec:cap}, we will discuss the privacy-utility trade off and provide an algorithm for finding the ``optimal cap''. Let $ G_i 
$ denote the $i-$th observer's personal gaze map (after applying cap). 
The aggregated gaze map of all $n$ observers in the database is denoted by $ G = \frac{1}{n}\sum_{i=1}^n  G_i$. Here, we normalize $G$ by the number of observers in order to compare the noise-level under different setups. To simplify notations, we use $\mathcal{G} = \left(G_1,\cdots,G_n\right)$ to denote the collection of all observers' gaze maps. Similarly, we use $\Gmi = \left(G_1,\cdots,G_{i-1},G_{i+1},G_n\right)$ to denote the collection of all observers' personal gaze maps except the $i-$th observer. Then, we will define  several gaze-map-aggregation mechanisms as follows:
\begin{itemize}
  \item $\calm_{\text{noise-free}}$: Directly output the aggregated gaze map. Formally,
  $\calm_{\text{noise-free}}( G_1,\cdots, G_n) =  G = \frac{1}{n}\sum_{i=1}^n  G_i$.
  \item $\calm_{\text{rs1}(c)}$: Randomly select $cn$ gaze maps from dataset (without replacement) and calculate aggregated gaze map accordingly. Formally, assuming the selected gaze maps are $ G_{j_1}\cdots,  G_{j_{cn}}$, $\calm_{\text{rs1}}( G_1,\cdots, G_n) =  G = \frac{1}{cn}\sum_{k=1}^{cn}  G_{j_k}$.
  \item $\calm_{\text{rs2}(c)}$: Similar with $\calm_{\text{rs1}(c)}$, the only difference is the sampling process is with replacement.
  \item $\calm_{\text{Gaussian}(\sigma_N)}$: Adding Gaussian noise with standard deviation (noise-level) $\sigma_N$ to all pixels independently. Formally, $\calm_{\text{Gaussian}(\sigma_N)}( G_1,\cdots, G_n) =  G +  \epsilon_{\sigma_N}$, where $ \epsilon_{\sigma_N}$ is a $r$ dimensional Gaussian noise term with zero mean and standard deviation $\sigma_N$ (all dimensions are mutually independent).
  \item $\calm_{\text{Laplacian}(\sigma_L)}$: Similar with $\calm_{\text{Gaussian}(\sigma_N)}$, the only difference is Laplacian noise with noise level $\sigma_L$ is added instead of Gaussian noise.
\end{itemize}
In short, $\calm_{\text{rs1}(c)}$ and $\calm_{\text{rs2}(c)}$ inject sampling noise to the output while $\calm_{\text{Gaussian}(\sigma_N)}$ and $\calm_{\text{Laplacian}(\sigma_L)}$ inject additive noise.

\subsection{Defining eye-tracking differential privacy}
\label{sec:dp2}
We start with re-phrasing the definition of $(\epsilon,\delta)-$differential privacy to eye tracking data. \textcolor{blue}{In the following discussion, we assume that the aggregated gaze map $G$ (or its noisy version) has been publicly released\footnote{Because DP focuses on worst case scenarios, the adversary also knows all other observers individual gazemaps.}.}
The goal of our research is to protect observers' personal gaze maps $ G_1,\cdots, G_n$ by adding appropriate noise to the aggregated gaze map. Using the notation in Section~\ref{subs:notation}, we assume that $\Gmi$, all gaze maps other than $ G_i$, are known by the adversary. For any set $\cals$ of outputting gaze maps, $(\epsilon,\delta)-$differential privacy is formally defined as follows.
\begin{definition}[$(\epsilon,\delta)-$DP]\label{def:dp2}
For any set of event $\cals$, any collection of gaze maps $\Gmi$ known by the adversary, we say a mechanism $\calm$ is $(\epsilon,\delta)-$differentially private if and only if
\begin{equation}\label{equ:dp2}
\Pr[\calm( G_i^*,\Gmi)\in \cals \mid \Gmi] \leq e^{\epsilon}\Pr[\calm( G_i^{**},\Gmi)\in \cals \mid \Gmi]+\delta,
\end{equation}
where $ G_i^*$ and $ G_i^{**}$ are any gaze maps of the $i-$th observer.
\end{definition}
According to differential privacy literatures \cite{dwork2014algorithmic}, there is no hard threshold between good and poor privacy. For the purpose of illustration, we define the following ``privacy levels'' in the remainder of this paper:
\begin{itemize}
  \item Poor privacy: $\delta = \Omega(1/n)$.
  \item Okay privacy: $\epsilon = 3$ and $\delta = n^{-3/2}$.
  \item Good privacy: $\epsilon = 1$ and $\delta = n^{-3/2}$.
\end{itemize}
Note ``okay privacy'' and ``good privacy'' are two examples we used for implementation. Practitioners can set their values of $\epsilon$ and $\delta$ according to their requirements (smaller $\epsilon$ and $\delta$ means better privacy). Note again $\delta = \Omega(1/n)$ is widely acknowledged as poor privacy \cite{dwork2014algorithmic}.  

\subsection{There is no free privacy}
\label{subs:nofree}
\textcolor{blue}{We first use $\calm_{\text{noise-free}}$ (poor privacy) as an example to connect intuition and the definition of DP.} Intuitively, if the adversary has the noiseless aggregated gaze map $ G$ and all other observers' gaze maps $\Gmi$, he/she can perfectly recover $ G_i$ by calculating $n G - \sum_{j\neq i}  G_j = \left(\sum_{j = 1}^n  G_j\right) - \sum_{j\neq i}  G_j =  G_i$.

Using Definition~\ref{def:dp2} and letting $ G_i^{*} =  G_i \neq  G_i^{**}$ and $\cals = \left\{ G\right\}$,
{\small $$\Pr[\calm( G_i^*,\Gmi)\in \cals \mid \Gmi] = 1\;\;\text{and}\;\;\Pr[\calm( G_i^{**},\Gmi)\in \cals \mid \Gmi] = 0,$$}
\noindent
because $ G$ will not be a possible output if $ G_i \neq  G_i^{*}$. Thus, we know $\delta$ can't be less than 1 to make Inequality~\ref{equ:dp2} hold. Considering $\delta = \Omega(1)$ corresponds to poor privacy, we know $\calm_{\text{noise-free}}$ has poor privacy in the language of $(\epsilon,\delta)-$DP defined in Definition~\ref{def:dp2}.

\subsection{Random selection gives poor privacy}
In Section~\ref{subs:notation}, we proposed two versions of random selection mechanisms. The first version ($\calm_{\text{rs1}}$) randomly selects $cn$ observers without replacement while the second version  ($\calm_{\text{rs2}}$) selects $cn$ with replacement.

\begin{theorem} [without replacement]\label{theo:rs1} Mechanism $\calm_{\text{rs1}}$ has poor privacy.
\end{theorem}

\begin{proof}
We prove $\calm_{\text{rs1}}$'s privacy by considering the following case: assuming resolution $r = 1$\footnote{This case also holds for $r>1$ because the first pixel already leaked information.}, all observers other than the $i-$th did not look at the only pixel, we have,
{\small
\begin{equation}\nonumber
\begin{split}
\Pr\left[\calm_{\text{rs1}}( G_1,\cdots,  G_n) =  \frac{1}{cn}\,\bigg|\, G_i = \textbf{1}, \Gmi = {0}\right] &= c\;\;\;\text{and}\\
\Pr\left[\calm_{\text{rs1}}( G_1,\cdots,  G_n) =  \frac{1}{cn}\,\bigg|\, G_i  = \textbf{0}, \Gmi = {0}\right] &= 0,
\end{split}
\end{equation}}
where $\Gmi = 0$ means all elements in collection $\Gmi$ equals to $0$. Thus, we know $\delta$ can't be less than $c$ to make (\ref{equ:dp2}) hold. Then, Theorem~\ref{theo:rs1} follows because $c = \Omega(1/n)$ ($cn = \Omega(1)$ is the number of observers selected).
\end{proof}

\begin{theorem} [with replacement]\label{theo:rs2} Mechanism $\calm_{\text{rs2}}$ has poor privacy.
\end{theorem}
Proof of Theorem~\ref{theo:rs2} (see Appendix~\ref{apd:thrors2} in Supplementary materials) is similar to the proof of Theorem~\ref{theo:rs1}. 

\subsection{Achieving good privacy with random noise}
In this section, we show that adding Gaussian or Laplacian noise can give good privacy if the noise level satisfies certain conditions based on user-defined privacy levels.

\subsubsection{Gaussian Noise}
Gaussian noise is widely used noise in many optical systems. In $\calm_{\text{Gaussian}(\sigma_N)}$, we add Gaussian noise with standard deviation $\sigma_N$ independently to all pixels of the aggregated gaze map. The probability density $p_N$ of outputting aggregated gaze map $ G^{(N)}$ is
{\small
\begin{equation}\label{equ:gaussian_lab}
\begin{split}
    &p_N\left(\calm_{\text{Gaussian}(\sigma_N)}( G_1,\cdots,  G_n) =  G^{(N)}\right)\\
    =\;& \frac{1}{\left(2\pi\sigma_N\right)^{r/2}}\cdot\exp\left(-\frac{\lnorm G^{(N)}- G\rnorm_2^2}{2\sigma_N^2}\right),
\end{split}
\end{equation}}
which is a $r$ dimensional Gaussian distribution such that all dimensions are independent. Note all $\ell_2$ norm in main paper and appendix represent Frobenius norm of matrices. 
For simplification, we use $p_N\left( G^{(N)}\right)$ to represent $p_N\left(\calm_{\text{Gaussian}(\sigma_N)}( G_1,\cdots,  G_n) =  G^{(N)}\right)$ when without ambiguity. The next Theorem shows announcing $ G^{(N)}$ ($\calm_{\text{Gaussian}(\sigma_N)}$'s output) will not give much information to adversary if the noise-level is as required (for any $(\epsilon,\delta)$, we can always find noise level $\sigma_N$ to guarantee $(\epsilon,\delta)-$DP).
\begin{theorem}[Gaussian Noise]\label{theo:gaussian1}
For any noise level $\sigma_N \geq \frac{m}{n\epsilon}\cdot\sqrt{r\left(\frac{\epsilon}{2}+\ln\frac{r}{\delta}\right)}$, $\calm_{\text{Gaussian}(\sigma_N)}$ is $(\epsilon,\delta)-$differentially private.
\end{theorem}
Theorem~\ref{theo:gaussian1} basically says we can always find a noise level $\sigma_N$ to meet any user-defined privacy level (any $\epsilon$ and $\delta$).
\begin{proof}
Let $ G_i^{*}$ and $ G_i^{**}$ to denote any two possible gaze maps of the $i-$th observer. To simplify notation, we use $\Gmiv = \frac{1}{n-1}\sum_{j\neq i}  G_{j}$ to denote the aggregated gaze map from observers other than the $i-$th. If the $i-$th observer's gaze map is $ G_i^*$, the probability density of the outputting $p_N( G^{(N)}\mid  G_i =  G_i^*)$ is 
{\footnotesize $$p_N( G^{(N)}\mid  G_i =  G_i^*) = \frac{1}{\left(2 \pi \sigma_N\right)^{r/2}}\exp\left(-\frac{1}{2\sigma_N^2} \lnorm\frac{ G_i^*}{n} + \frac{n-1}{n}\Gmiv -  G^{(N)}\rnorm_2^2\right),$$}
Similarly, if the $i-$th observer's gaze map is $ G_i^{**}$, we have,
{\footnotesize$$p_N( G^{(N)}\mid  G_i =  G_i^{**}) = \frac{1}{\left(2 \pi \sigma_N\right)^{r/2}}\exp\left(-\frac{1}{2\sigma_N^2} \lnorm\frac{ G_i^{**}}{n} + \frac{n-1}{n}\Gmiv -  G^{(N)}\rnorm_2^2\right),$$}
For any $ G_i^{*}$,  $ G_i^{**}$ and  $\Gmiv$, we have,
{\footnotesize \begin{equation}\nonumber
\begin{split}
&\frac{p_N( G^{(N)}\mid  G_i =  G_i^{**})}{p_N( G^{(N)}\mid  G_i =  G_i^{*})}\\
=\;& \exp\left(\frac{1}{2\sigma_N^2}\cdot\left(\lnorm\frac{ G_i^{*}}{n} + \frac{n-1}{n}\Gmiv -  G^{(N)}\rnorm_2^2 - \lnorm\frac{ G_i^{**}}{n} + \frac{n-1}{n}\Gmiv -  G^{(N)}\rnorm_2^2\right)\right)\\
\leq\;& \exp\left(\frac{2\lnorm\frac{ G_i^*}{n} + \frac{n-1}{n}\Gmiv -  G^{(N)}\rnorm_2\cdot\lnorm G_i^{**}- G_i^*\rnorm_2 + \lnorm G_i^{**}- G_i^{*}\rnorm_2^2}{2\sigma_N^2}\right).
\end{split}
\end{equation}}
Letting $ \mu = \frac{ G_i^*}{n} + \frac{n-1}{n}\Gmiv$ and considering $\lnorm G_i^{**}- G_i^*\rnorm_2 \leq \frac{m\sqrt{r}}{n}$, we have,
{\small
$$\frac{p_N( G^{(N)}\mid  G_i =  G_i^{**})}{p_N( G^{(N)}\mid  G_i =  G_i^{*})} \leq \exp\left(\frac{\frac{2m\sqrt{r}}{n}\lnorm  G^{(N)}-\mu \rnorm_2 + \frac{m^2r}{n^2}}{2\sigma_N^2}\right).$$}
Thus, for any $  G^{(N)}$ such that $\lnorm G^{(N)}-\mu\rnorm_2 \leq \frac{n}{m\sqrt{r}}\epsilon\sigma^2-\frac{m\sqrt r}{2n}$, the $\epsilon$ requirement of DP is always met. Then, we bound the tail probability for all cases where $\epsilon$'s requirement is not met.
{\small
\begin{equation}\nonumber
\begin{split}
&\Pr\left[\lnorm G^{(N)} - \mu\rnorm_2  > \frac{n}{m\sqrt{r}}\epsilon\sigma_N^2-\frac{m\sqrt r}{2n}\right]\\
\leq\;& \sum_{j=1}^r \Pr\left[\left| G_j^{(N)} - \mu_j\right|  > \frac{n}{m\sqrt{r}}\epsilon\sigma_N^2-\frac{m\sqrt r}{2n}\right]\\
\leq\;& r\cdot\exp\left(-\frac{n^2\epsilon^2\sigma_N^2}{2m^2r} + \frac{\epsilon}{4}\right).\\
\end{split}
\end{equation}}
When $\sigma_N \geq \frac{m}{n\epsilon}\cdot\sqrt{r\left(\frac{\epsilon}{2}+\ln\frac{r}{\delta}\right)}$, we have,
{\small
$$\Pr\left[\lnorm G^{(N)} - \mu\rnorm_2  > \frac{n}{m\sqrt{r}}\epsilon\sigma_N^2-\frac{m\sqrt r}{2n}\right] \leq r\cdot\exp\left(-\frac{n^2\epsilon^2\sigma_N^2}{2m^2r} + \frac{\epsilon}{4}\right) \leq \delta$$}
Then, Theorem~\ref{theo:gaussian1} follows by the definition of $(\epsilon,\delta)-$DP.
\end{proof}

\subsubsection{Laplacian Noise}
Laplacian noise is the most widely used in many differential privacy problems. However, we will show Laplacian noise is not as suitable as Gaussian noise for protecting eye tracking data. The next Theorem shows 
$ G^{(L)}$ will not give much information to the adversary if the noise-level is as required.
\begin{theorem}[Laplacian Noise]\label{theo:lpls1}
Using the notations above, for any $\sigma_L \geq \frac{\sqrt{2}\cdot mr}{\epsilon n}$, $\calm_{\text{Laplacian}(\sigma_L)}$ is $(\epsilon,0)-$differentially private.
\end{theorem}
Proof of Theorem~\ref{theo:lpls1} (see Appendix~\ref{apd:throlpls1} in supplementary material) is very similar with Theorem~\ref{theo:gaussian1}.
However, the required noise level, $\sigma_L \geq \frac{\sqrt{2}\cdot mr}{\epsilon n}$, normally is much higher than the requirement of Gaussian noise, $\sigma_N \geq \frac{m}{n\epsilon}\cdot\sqrt{r\left(\frac{\epsilon}{2}+\ln\frac{r}{\delta}\right)}$.  One can see the Laplacian mechanism requires one more $\sqrt r$ term on noise level, which normally corresponds to $\sim 10^2$ times higher noise level.

\begin{figure*}[!ht]
     \subfloat[Surfaces for a chosen privacy level.\label{subfig-1:noisesurface}]{%
       \includegraphics[width=0.24\textwidth]{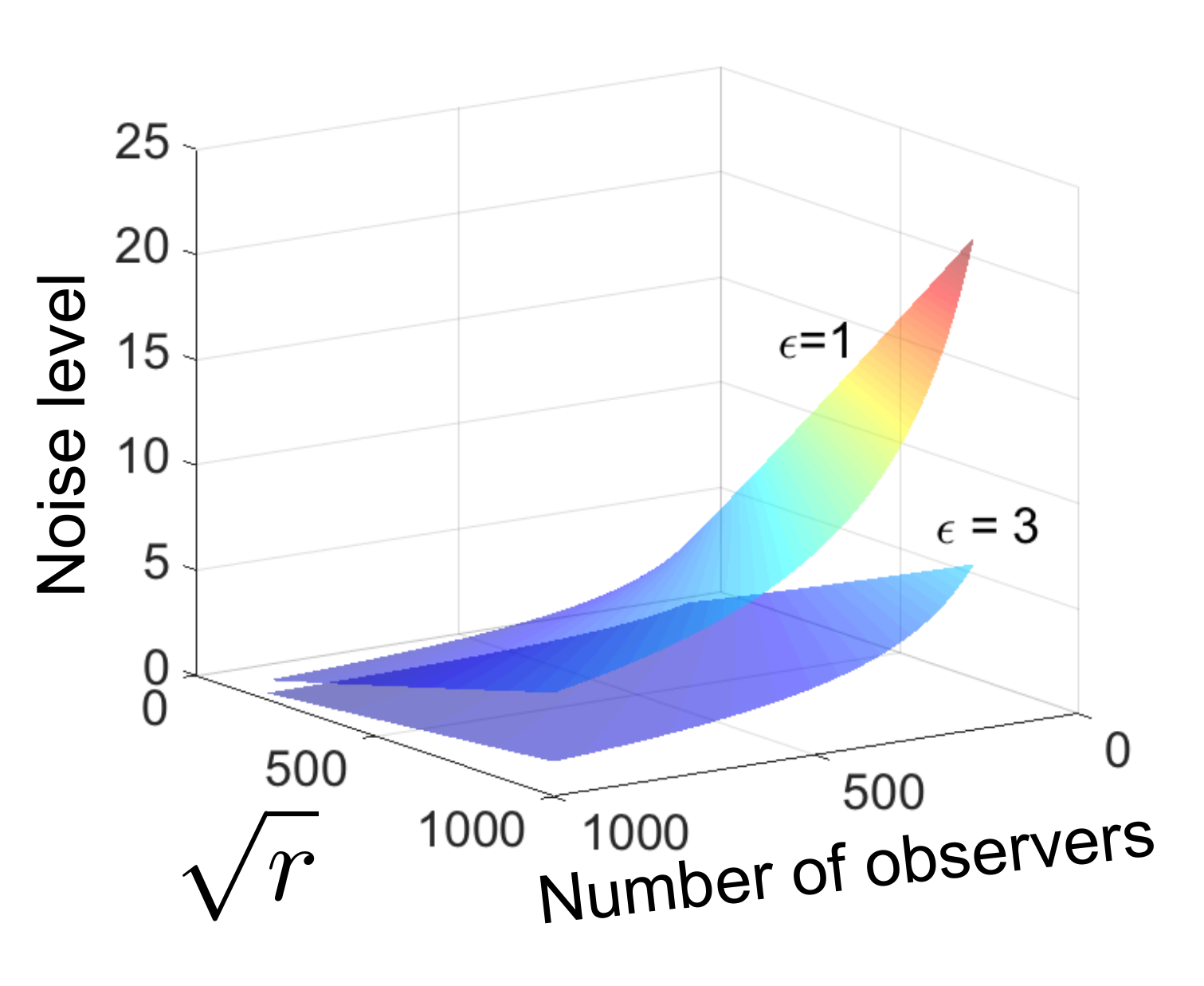}
     }
     \subfloat[Slice of the surface for $300\times300$ images. 
     \label{subfig-2:noisesurface}]{%
       \includegraphics[width=0.24\textwidth]{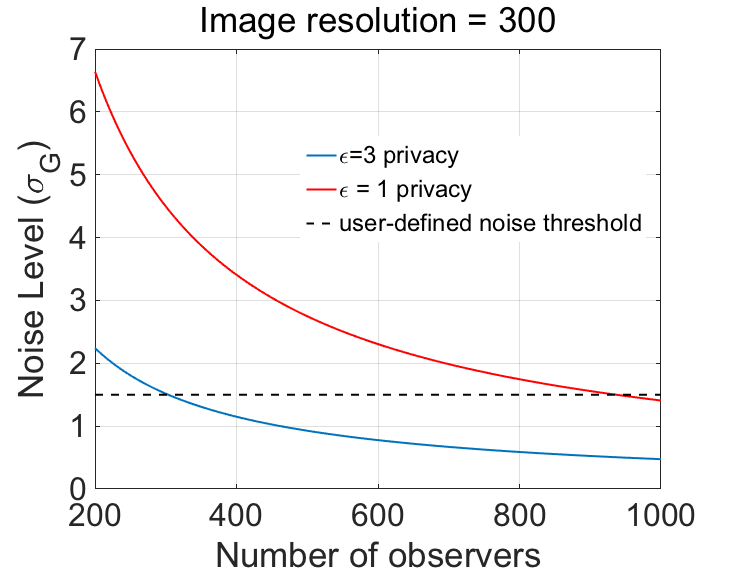}
     }
     \subfloat[Original heatmap.\label{subfig-3:noisesurface}]{%
       \includegraphics[width=0.24\textwidth]{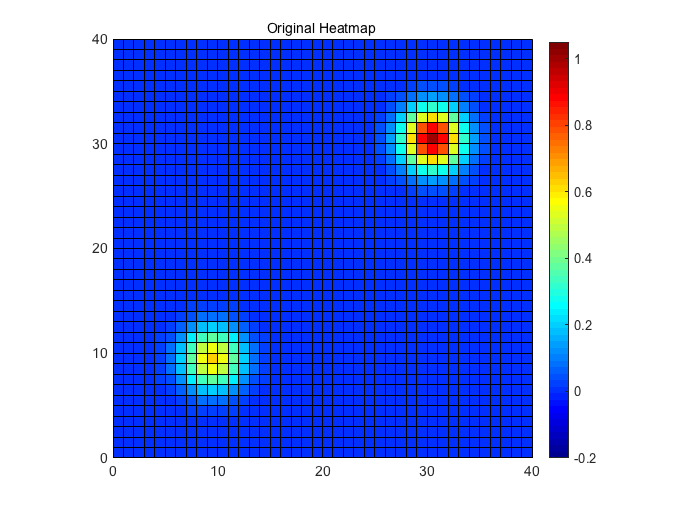}
     }
     \subfloat[Privacy enhanced heatmap for $\epsilon=1$ level.\label{subfig-4:noisesurface}]{%
       \includegraphics[width=0.24\textwidth]{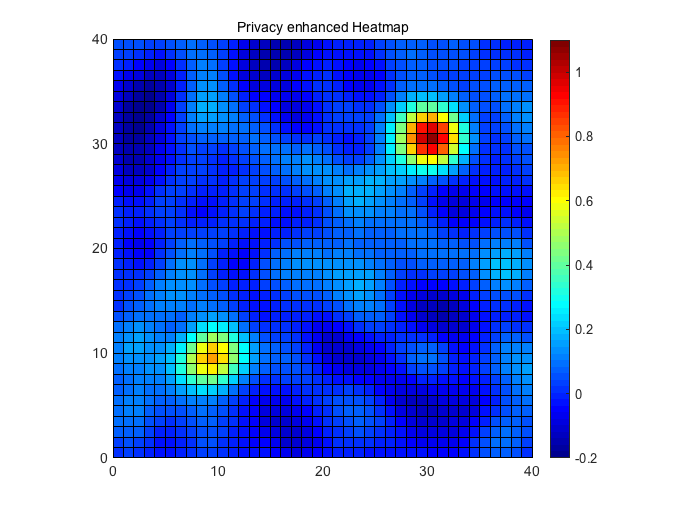}
     }
     \caption{We examine the privacy-utility tradeoff for selected values of $\sigma_N$ for a simulated heatmap. The greater the noise level we choose to add, the stronger is the privacy guarantee. The relevant stakeholders decide what level of noise is acceptable for a given application. For example, in Figure~\ref{subfig-4:noisesurface}, the hotspots are still clear, and a UX designer may find this acceptable for the purpose of getting feedback on the design of a website.}
     \label{fig:noisesurface}
\end{figure*}


\section{Privacy-utility tradeoff}\label{sec:put}
According to Theorem~\ref{theo:gaussian1} and Theorem~\ref{theo:lpls1}, we know better privacy (smaller $\epsilon$ and $\delta$) usually requires higher noise level. In this section, we will conduct experiments to show how Gaussian and Laplacian noise influence the utility, i.e., the corresponding heatmap.

\subsection{Noise level vs. information loss}
In Figure~\ref{fig:noisesurface}(a), we show a three-dimensional plot where the x and y axes are $\sqrt{r}$ and $n$ respectively. The reader may revisit notations in Section~\ref{subs:notation}. On the vertical z-axis, we plot $\sigma_N$, specifically based on the formula given by Theorem~\ref{theo:gaussian1}. The upper surface shows $\sigma_N$ for good privacy ($\epsilon = 1$ and $\delta = n^{-3/2}$).  The lower surface shows for okay privacy ($\epsilon = 3$ and $\delta = n^{-3/2}$).  Any value of $\sigma_N$ above this surface will provide okay privacy, and any value above the upper surface will provide good privacy.

In Figure~\ref{fig:noisesurface}\,(a), as the image resolution increases, a larger number of observers is needed in the dataset to maintain the guarantee of good privacy. If there is a small number of observers, good privacy can be achieved by downsampling the image. In Figure~\ref{fig:noisesurface}\,(b) we show a slice of this surface at $\sqrt{r}=300$. The dotted lines show an example noise level that we could have set based on what we find acceptable for utility. This is of course user-defined, and will vary depending on the application. The graphs illustrate that at a selected noise level, e.g., $\sigma_N=1.5$, we can achieve good privacy for a $300\times300$ image if we have of the order of $n=900$ observers. For a dataset that has $n=300$ observers, we can tell the participants that we can achieve Okay privacy. We show two simulated heatmaps in Figure~\ref{fig:noisesurface}\,(c) and\,(d). The location of the hotspots is unchanged for all practical purposes in the noisy but private heatmap.

We quantify the privacy-utility tradeoff in Figure~\ref{fig:tradeoff_plot}. 100 noisy heatmaps are generated using the workflow in Figure~\ref{fig:workflow}. Real-world $1050\times1680$ gaze maps from five observers looking at a portrait of a woman are used here.\nolinebreak\footnote{%
    Gaze data from dataset of
    \citeN{raiturkar2018},
    stimulus image from
    \citeN{farid2012perceptual} and \citeN{mader2017identifying}.
}
The original heatmap is shown in Figure~\ref{fig:workflow} to the right. For the purpose of the noisy heatmap, we assume the number of observers in dataset  is $50,000$\footnote{If the number of observer is much smaller than 50,000, the practitioner could either down-sample gaze maps or sacrifice privacy (setting larger $\epsilon$ and $\delta$) to get an acceptable noise level.} (the noise is added according to $n = 50,000$ and Theorem~\ref{theo:gaussian1} and Theorem~\ref{theo:lpls1}). We simulate this large number of observers by replicating each of the five real observers $10000$ times. 

In the supplementary materials, we show the original heatmap overlaid on the stimulus image in high resolution (original.png). We also show examples of privacy enhanced heatmaps for this original heatmap at the $\epsilon=1.5$ privacy level (privancyenhanced.mpg). For this image resolution, $\sigma_G=0.0986$ based on Theorem~\ref{theo:lpls1}.

We numerically analyzed \emph{correlation coefficient} (CC) and \emph{mean square error} (MSE) of noisy heatmaps under different privacy levels (different values of $\epsilon$ while fixing $\delta = n^{-3/2}$). The cap $m = 1$ is decided according to Algorithm~\ref{alg:findm} (see Section~\ref{sec:cap} for details). 100 noisy heatmaps are generated under each setting. The average CC and MSE of those generated noisy heatmaps are plotted in Figure~\ref{fig:tradeoff_plot}. Error bars in Figure~\ref{fig:tradeoff_plot} represent the standard deviations.

It can be seen from the Figure~\ref{fig:tradeoff_plot} that Laplacian mechanism results in much more information loss than Gaussian mechanism to achieve same level of privacy under our setting. For both Gaussian and Laplacian mechanisms, one can see that better privacy (smaller $\epsilon$) usually means more information loss in the outputting heatmap. 

We note that these graphs are based on real data of only five observers on one stimulus image. This graph is an example of how a practitioner may visualize the privacy-utility tradeoff in any given application domain. In practice, stakeholders would use our proposed workflow on their dataset to prepare such visualizations for different settings of the internal parameters ($m,\epsilon,\delta$) to help them evaluate the privacy-utility tradeoff. We note also that Theorem~\ref{theo:gaussian1} is specific to aggregate heatmaps. For any other mechanism, the appropriate theorems would need to be worked out and the workflow modified to be consistent with the problem definition. We also point out that while mean squared error and cross-correlation are readily computed, they do not fully reflect the information lost or retained when noise is added. As an illustration, in Figure~\ref{fig:noisesurface}, the hotspots in the privacy enhanced heatmap are still clear, and a UX designer may find that the heatmap acceptable for their use case even though the MSE and CC metrics suggest otherwise.

\begin{figure}
    \centering
    \includegraphics[width=0.5\textwidth]{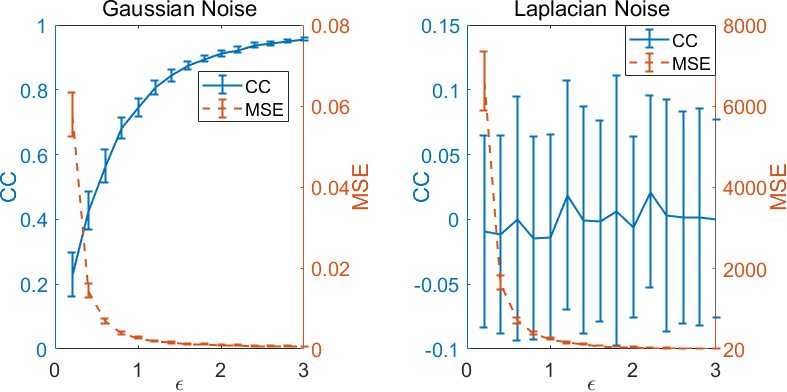}
    \caption{Similarity between the privacy enhanced heatmap and original heatmap when $\epsilon$ is varied. The smaller the value of $\epsilon$ the stronger is the privacy guarantee from the DP perspective. This graph illustrates the privacy-utility tradeoff: as $\epsilon$ is made smaller, the mean squared error increases and the cross-correlation decreases. The Laplacian mechanism results in lesser similarity than the Gaussian mechanism.}
    \label{fig:tradeoff_plot}
\end{figure}

\subsection{Computing the optimal ``cap''}\label{sec:cap}
In order to achieve better privacy with less information loss, we set a cap on the maximum number of times an observer's gaze position falls on a pixel. This cap was denoted by $m$ in Section~\ref{sec:theoretical}. Here we discuss the information loss on different settings on $m$.

When $m$ is larger, higher noise level is required to get the same privacy (both the upper bound for $\sigma_N$ and $\sigma_L$ are proportional to $m$). However, larger $m$ also corresponds to less information loss on every observer's gaze map. In other words, there is tradeoff between variance (noise) and bias (cap) on cap. Let $ G^{(N,m,\sigma^*)}$ to denote the gaze map outputted by Gaussian mechanism with cap $m$ and noise level $ \sigma_N = m\sigma^*$. Thus, $ G^{(N,\infty,0)}$ denotes the original aggregated gaze map and $ G^{(N,m,0)}$ denotes the aggregated gaze map with cap $m$ without adding any Gaussian noise. Algorithm~\ref{alg:findm} provides an implementable way to choose the best value of $m$ to optimize \emph{mean square error} (MSE).
\begin{algorithm}[htp]
\caption{Utility optimization algorithm on choosing $m$}\label{alg:findm}
\textbf{Input:} {Individual gaze maps $ G_1,\cdots, G_n$ and noise factor $\sigma^{*}$.}\\
\textbf{Initialization:} Calculate $ G^{(N,\infty,0)}$ and the maximum number of times one observer look at one pixel: $g_{\max} = \max_{i,j}\left[G_i(j)\right]$.\\
\For {$m = 1, \cdots, g_{\max}$}
{
    Calculate $\E\left[\text{MSE}\left( G^{(N,m,\sigma^*)}\right)\right]$ according to Theorem~\ref{theo:m}.
}
\textbf{Output:} {$m$ value with the smallest expected MSE.}\\
\end{algorithm}
In the next theorem, we will analytically analyze the expectation of MSE.
\begin{theorem}\label{theo:m}
The expected MSE of Gaussian mechanism with cap $m$ and noise level $\sigma_N = m\sigma^*$ is
{\small$$\E\left[\text{MSE}\left( G^{(N,m,\sigma^*)}\right)\right] = m^2\sigma^{*2}+\frac{1}{r} \sum_{j=1}^r \left(G^{(N,m,0)}(j) - G^{(N,\infty,0)}(j)\right)^2.$$}
\end{theorem}
\begin{proof}
By the definition of MSE, we have,
$$\E\left[\text{MSE}\left( G^{(N,m,\sigma^*)}\right)\right] = \frac{1}{r}\sum_{j=1}^r\E\left[{\left(G^{(N,m,\sigma^*)}(j) - G^{(N,\infty,0)}(j)\right)^2}\right]$$
Using the notations defined above, the expected square error on $j-$th pixel is 
{\footnotesize
\begin{equation}\label{equ:m1}
\begin{split}
&\E\left[\left(G^{(N,m,\sigma^*)}(j) - G^{(N,\infty,0)}(j)\right)^2\right]\\
=\;&\E\left[\left(G^{(N,m,\sigma^*)}(j)\right)^2\right] - 2G^{(N,\infty,0)}(j)\cdot\E\left[G^{(N,m,\sigma^*)}(j)\right] + \left(G^{(N,\infty,0)}(j)\right)^2.
\end{split}
\end{equation}}
because $G^{(N,m,\sigma^*)}_j \sim \caln\left(G_j^{(N,m,0)},m^2\sigma^{*2}\right)$, we have,
\begin{equation}\label{equ:m2}
\begin{split}
\E\left[\left(G^{(N,m,\sigma^*)}(j)\right)\right] &= G^{(N,m,0)}(j)\;\;\;\text{and}\\
\E\left[\left(G^{(N,m,\sigma^*)}(j)\right)^2\right] &= \sigma^{*2}m^2 + \left(G^{(N,m,0)}(j)\right)^2.
\end{split}
\end{equation}
Theorem~\ref{theo:m} follows by combining (\ref{equ:m1}) and (\ref{equ:m2}).
\end{proof}

Then, we analyze the complexity of Algorithm~\ref{alg:findm} in the next theorem, which says Algorithm~\ref{alg:findm} is with linear-time complexity.
\begin{theorem}\label{theo:complex}
The complexity of Algorithm~\ref{alg:findm} is $O\left(g_{\max}\cdot nr\right)$.
\end{theorem}
\begin{proof}
Rewriting the expected MSE in Theorem~\ref{theo:m}, we have,
$$\E\left[\text{MSE}\right] = m^2\sigma^{*2}+\frac{1}{r}\lnorm G^{(N,m,0)} - G^{(N,\infty,0)} \rnorm_2^2,$$ where the $\ell_2$ norm still represents Frobenius norm. Since $ G^{(N,m,0)}$ and $ G^{(N,\infty,0)}$ are $r$-dimensional vectors, the complexity of computing expected MSE for a given $m$ and $ G^{(N,m,0)}$ will be $O(m)$.

Then, we evaluate the complexity of calculating the capped noise-free aggregated gaze map $ G^{(N,m,0)}$. Since we are adding cap to each observer's individual gaze map, we can add cap to every pixel of all observers. Thus, one can see there are $nr$ pixels from $n$ observers in total. Considering the for loop in Algorithm~\ref{alg:findm} runs $g_{\max}$ times, Theorem~\ref{theo:complex} follows.
\end{proof}

\section{Implications}

\subsubsection{Datasets are growing.}
In contrast to the previous research paradigm where datasets were collected, archived, and then released, there is a growing trend to crowd-source data collection, via mobile apps for example, so that new data is continually being added to the dataset. With the methods presented, the new data is safe as long as the publicly available dataset is put through the Gaussian noise mechanism.
Another way that eye tracking datasets might seek to preserve a user's privacy is by releasing their eye movements, but not \textit{what} they were looking at. With the methods we present, releasing the stimulus image/video that observers look at is safe because even in the worst case an adversary will not be able to guess at what a particular individual looked at.

\subsubsection{Why can the generic theorem of differential privacy not be applied to eye tracking?}
Unlike classical databases, every observer in eye tracking database contributes much richer information  (i.e., millions of pixels) than individuals in classical databases. However, the generic theorems in differential privacy do not focus on high-dimensional data. Simply applying union bounds will result in very loose privacy bounds and unacceptable noise levels.

\subsubsection{Why are we adding noise when the field is spending so much time and effort removing it?} 
There has been much research in eye tracking to improve the accuracy of eye tracking to maximize the utility and applicability of eye tracking devices for diverse use cases. This work has been directed at sources of noise that are inherent to the process, such as
sensor and measurement noise. However, as eye tracking becomes ubiquitous, there is a cost for the individual user whose data is being recorded and for the organizations who are safeguarding and distributing this data. This cost is the sacrifice of privacy of the individual. We do not argue for reversing the technological push towards reliable, accurate eye tracking. Rather, we argue that our objective as a community must expand to include privacy in addition to utility. For those situations where privacy is deemed to be worth protecting, we introduce flexible mechanisms to do so. Noise is added to data in the aggregate, not to any individual's data. Further, the noise function is fully understood, and its parameters are set based on the desired privacy-utility tradeoff. Unlike measurement noise, whose source may not be fully understood, we add noise in a controlled and measured way to achieve a specific objective. 

\subsubsection{Why should the research community care?}
This research requires an interdisciplinary approach. The eye tracking community cannot just ``leave it to the privacy researchers'' because the theoretical guarantees that form the basis of this framework are highly dependent on the particular mechanisms that the data goes through (the functional forms in the equations, the particular thresholds, etc.). These mechanisms have to be developed collaboratively to preserve the utility of the output for eye tracking applications. 

\subsubsection{Why should the industry care?}
The push towards ubiquitous eye tracking is being driven by large investments by major industry players. While their applications are highly data-dependent, their customers are increasingly data-sensitive. This paper proposes the first of a class of solutions which pair theoretical analysis from a DP-perspective with a practically implementable workflow for developers. This work opens the door for a responsible industry that can inform their users that while they may eye track the users at very high accuracy and resolution to enable foveated rendering (for example), they would put this data through mechanism A or B before releasing it to the app developers. 

\vspace{-5mm}
\textcolor{blue}{
\section{Conclusions and Future directions}}
\textcolor{blue}{We have proposed to apply the notion of differential privacy toward the analysis of privacy of eye-tracking data.} We have analyzed the privacy guarantees provided by the mechanisms of random selection, and additive noise (Gaussian and Laplacian noise). \textcolor{blue}{The main takeaway from this paper is that adding Gaussian noise will guarantee differential privacy; the noise level should be appropriately selected based on the application.} Our focus is on static heatmaps as a sandbox to understand how the definitions of differential privacy apply to eye tracking data. In this sense, this paper is a proof of concept. Eye tracking data is fundamentally temporal in nature, and the privacy loss if an adversary could access saccade velocities and dynamic attention allocation would be much greater than static heatmaps. Future work would systematically consider all the different ways in which eye tracking data is analyzed and stored.

We have considered two noise models (Gaussian and Laplacian noise). Follow up work might consider the privacy-utility trade-off for different noise models like pink noise. For temporal data such as raw eye movements, it may even be relevant to understand which noise models are more realistic. In other words, if the user's virtual avatar was driven by privacy enhanced eye tracking data, it should still appear realistic and natural. 

The mechanisms and analyses presented here apply to real-valued data that can be aligned to a grid and capped to a maximum value without loss of utility. Though our focus has been on eye tracking heatmaps, there are other data that fall in this category, for example, gestures on a touchscreen, or readings from a force plate. It would also be interesting to generalize these mechanisms and analyses to other physiological data such as heart rate, galvanic skin response, and even gestures or gait. These data are conceptually similar to eye tracking data in that they carry signatures of the individual's identity and markers of their health and well-being.
{Furthermore, in physiological domains many data and analyses are temporal in nature.
It would be
interesting and important
to define and analyze differential privacy for temporal data.
}
 
\mbox{}\\[1ex]
\begin{acks}
This material is based upon work supported by the National Science Foundation under Grant No. 1566481. Any opinions, findings, and conclusions or recommendations expressed in this material are those of the author(s) and do not necessarily reflect the views of the National Science Foundation.
\end{acks}

\mbox{}\\[1ex]
\balance

\clearpage
\appendix
\setcounter{section}{0}
\section{Missing proofs for theorems}\label{apd:thrors2}
\subsection{Proof and discussion for Theorem~\ref{theo:rs2}}
Theorem~\ref{theo:rs2} says $\calm_{\text{rs2}}$ has poor privacy. 

\begin{proof}
Considering the same case as the proof of Theorem~\ref{theo:rs1}, we have,
{\small
\begin{equation}\nonumber
    \begin{split}
&\Pr\left[\calm_{\text{rs2}}( G_1,\cdots,  G_n) \geq  \frac{1}{cn}\,\bigg|\, G_i = \textbf{0},  \Gmi = {0}\right] = 0\;\;\;\text{and}\\
&\Pr\left[\calm_{\text{rs2}}( G_1,\cdots,  G_n) \geq  \frac{1}{cn}\,\bigg|\, G_i = \textbf{1},  \Gmi = {0}\right] \\
=\;& 1- \left(1-\frac{1}{n}\right)^{cn} \approx 1-e^{-c} = \Theta(c) 
    \end{split}
\end{equation}}

Thus, we know $\delta$ can't be less than $c$ to make Inequality~\ref{equ:dp2} hold and Theorem~\ref{theo:rs2} follows.
\end{proof}

\subsection{Proof for Theorem~\ref{theo:lpls1}}\label{apd:throlpls1}
\begin{proof}
The probability density function $p_L$ of output $ G^{(L)} = (G^{(L)}(1).\cdots,G^{(L)}(r))$ is
{\small
\begin{equation}\nonumber
    \begin{split}
    &p_L\left(\calm_{\text{Laplacian}(\sigma_L)}( G_1,\cdots,  G_n) =  G^{(L)}\right)\\
    =\;& \frac{1}{\left(\sqrt 2 \sigma_L\right)^r}\cdot \exp\left(-\frac{\sqrt 2 }{\sigma_L} \lnorm G -  G^{(L)}\rnorm_1\right).
    \end{split}
\end{equation}}
To simplify notations, we use $p_L\left( G^{(L)}\right)$ to represent\\ $p_L\left(\calm_{\text{Laplacian}(\sigma_L)}( G_1,\cdots,  G_n) =  G^{(L)}\right)$ when without ambiguity. 
Let $ G_i^{*}$ and $ G_i^{**}$ to denote any two possible gaze maps of the $i-$th observer. If the $i-$th observer's gaze map is $ G_i^*$, the probability density function of the outputting $p_{L}( G^{(L)}\mid  G_i =  G_i^*)$ is
{\footnotesize$$p_{L}( G^{(L)}\mid  G_i =  G_i^*) = \frac{1}{\left(\sqrt 2 \sigma_L\right)^r}\exp\left(-\frac{\sqrt 2}{\sigma_L} \lnorm\frac{ G_i^*}{n} + \frac{n-1}{n} G_{-i} -  G^{(L)}\rnorm_1\right),$$}
where we abused notation to let $ G_{-i} = \frac{1}{n-1}\sum_{j\neq i} G_{j}$, which is the aggregated gaze map except $ G_i$.
Similarly, if the $i-$th observer's gaze map is $ G_i^{**}$, we have,
{\footnotesize $$p_{L}( G^{(L)}\mid  G_i =  G_i^{**}) = \frac{1}{\left(\sqrt 2 \sigma_L\right)^r}\exp\left(-\frac{\sqrt 2}{\sigma_L} \lnorm\frac{ G_i^{**}}{n} + \frac{n-1}{n} G_{-i} -  G^{(L)}\rnorm_1\right).$$}
For any $ G_i^{*}$,  $ G_i^{**}$ and  $ G_{-i}$, we have,
{\footnotesize{
\begin{equation}\nonumber
\begin{split}
&\frac{p_{L}( G^{(L)}\mid  G_i =  G_i^{**})}{p_{L}( G^{(L)}\mid  G_i =  G_i^{*})}\\
=\;& \exp\left(\frac{\sqrt 2}{\sigma_L}\cdot\left(\lnorm\frac{ G_i^{*}}{n} + \frac{n-1}{n} G_{-i} -  G^{(L)}\rnorm_1 - \lnorm\frac{ G_i^{**}}{n} + \frac{n-1}{n} G_{-i} -  G^{(L)}\rnorm_1\right)\right)\\
\leq\;& \exp\left(\frac{\sqrt 2 \cdot || G_i^{**}- G_i^*||_1}{\sigma_L n}\right) \leq \exp\left(\frac{\sqrt {2}\cdot m r}{\sigma_L n}\right).\\
\end{split}
\end{equation}}}
Since the probability is the integral of PDF, the above upper bound for PDF ratio is also an upper bound for probability ratio. Thus, for any possible output set $S$, we have,
\begin{equation}\nonumber
\begin{split}
&\Pr\left[\calm_{\text{Laplacian}(\sigma_L)}( G_i^*, \Gmi)\in S \mid  \Gmi\right]\\
\leq\;& \exp\left(\frac{\sqrt {2}\cdot m r}{\sigma_L n}\right)\Pr\left[\calm_{\text{Laplacian}(\sigma_L)}( G_i^{**}, G_{-i})\in S \mid  \Gmi\right].
\end{split}
\end{equation}
and Theorem~\ref{theo:lpls1} follows by applying Definition~\ref{def:dp2}.
\end{proof}



\end{document}